\documentclass[10pt,journal,twocolumn,twoside]{IEEEtran} 
\usepackage{graphicx}
\usepackage{epstopdf}
\usepackage{float}
\usepackage{algorithmic}
\usepackage{array}
\usepackage{amsmath}
\usepackage{amssymb}
\usepackage{mdwmath}
\usepackage{mdwtab}
\usepackage{eqparbox}
\usepackage{stfloats}
\usepackage{fixltx2e}
\usepackage{hyperref}
\usepackage{cleveref}
\usepackage{booktabs}
\usepackage{multirow}
\usepackage{makecell}
\usepackage{subfigure}

\hypersetup{hypertex=true,
            colorlinks=true,
            linkcolor=blue,
            anchorcolor=blue,
            citecolor=blue}
\usepackage{cases} 

\usepackage{upgreek}

\usepackage{xcolor}
\usepackage{makecell}
\usepackage{amsmath}
\usepackage[boxed,ruled,commentsnumbered]{algorithm2e}
\usepackage{url}
\usepackage{cite}
\ifCLASSOPTIONcompsoc
\usepackage[caption=false,font=normalsize,labelfont=sf,textfont=sf]{subfig}
\else

\allowdisplaybreaks[4]

\newtheorem{assumption}{\bf Assumption}
\newtheorem{corollary}{\bf Corollary}
\newtheorem{theorem}{\bf Theorem}

\newtheorem{lemma}{\bf Lemma}

\newtheorem{remark}{\bf Remark}

\newtheorem{proof}{Proof}

\makeatletter

\renewcommand*{\@opargbegintheorem}[3]{\trivlist
      \item[\hskip \labelsep{\bfseries #1\ #2}] \textbf{(#3):}\ }
\makeatother

\begin{document}

\title
{OptiVote: Non-Coherent FSO Over-the-Air Majority Vote for Communication-Efficient Distributed Federated Learning in Space Data Centers}
\author{Anbang Zhang, Chenyuan Feng,~\IEEEmembership{Member, IEEE}, Wai Ho Mow,~\IEEEmembership{Senior Member, IEEE}, Jia Ye,~\IEEEmembership{Member, IEEE},\\
Shuaishuai Guo,~\IEEEmembership{Senior Member, IEEE}, Geyong Min,~\IEEEmembership{Member, IEEE}, and Tony Q. S. Quek,~\IEEEmembership{Fellow, IEEE} 
\thanks{This work was supported  by the Shandong Provincial Natural Science Foundation under Grant ZR2025MS983.
(\emph{*Corresponding author: Chenyuan Feng})

Anbang Zhang, Shuaishuai Guo are with School of Control Science and Engineering, Shandong University, Jinan 250061, China (e-mail: zab\_0613@163.com, shuaishuai\_guo@sdu.edu.cn);
Chenyuan Feng and Geyong Min are with the Department of Computer Science, University of Exeter, U.K. (email: c.feng@exeter.ac.uk, g.min@exeter.ac.uk).
Wai Ho Mow is with the Department of ECE, The Hong Kong University of Science and Technology, Hong Kong (e-mail: eewhmow@ust.hk);
Jia Ye is with School of Electrical Engineering, Chongqing University, Chongqing 400044, China. (email: jia.ye@cqu.edu.cn).
T.~Q.~S.~Quek is with the Information Systems Technology and Design Pillar, Singapore University of Technology and Design, Singapore 487372 (e-mail: tonyquek@sutd.edu.sg).
}}
\maketitle
\maketitle

\begin{abstract}
The rapid deployment of mega-constellations is driving the long-term vision of space data centers (SDCs), where interconnected satellites form in-orbit distributed computing and learning infrastructures. Enabling distributed federated learning in such systems is challenging because iterative training requires frequent aggregation over inter-satellite links that are bandwidth- and energy-constrained, and the link conditions can be highly dynamic. Communication-efficient aggregation therefore becomes a key enabler for scalable in-orbit intelligence.
In this work, we exploit over-the-air computation (AirComp) as an in-network aggregation primitive. However, conventional coherent AirComp relies on stringent phase alignment, which is difficult to maintain in space environments due to satellite jitter and Doppler effects. To overcome this limitation, we propose OptiVote, a robust and communication-efficient non-coherent free-space optical (FSO) AirComp framework for federated learning toward Space Data Centers. OptiVote integrates sign stochastic gradient descent (signSGD) with a majority-vote (MV) aggregation principle and pulse-position modulation (PPM), where each satellite conveys local gradient signs by activating orthogonal PPM time slots. The aggregation node performs MV detection via non-coherent energy accumulation, transforming phase-sensitive field superposition into phase-agnostic optical intensity combining, thereby eliminating the need for precise phase synchronization and improving resilience under dynamic impairments.
To mitigate aggregation bias induced by heterogeneous FSO channels, we further develop an importance-aware, channel state information (CSI)-free dynamic power control scheme that balances received energies without additional signaling. We provide theoretical analysis by characterizing the aggregate error probability under statistical FSO channels and establishing convergence guarantees for non-convex objectives. Extensive experiments demonstrate that OptiVote consistently outperforms representative baselines in both communication efficiency and learning accuracy, highlighting its potential for scalable and resilient distributed intelligence in future communication-constrained SDCs.
\end{abstract}


\section{Introduction}
\IEEEPARstart{T}{he} rapid deployment of large-scale low Earth orbit (LEO) satellite constellations is reshaping the architecture of future communication and computing systems \cite{11271330,9210567}. Beyond serving as transparent relays, modern satellites are increasingly equipped with high-performance onboard processors and advanced sensing capabilities, enabling them to operate as interconnected {space data centers} (SDCs) \cite{11003120,10935306,11293041}. In this emerging paradigm, satellites collaboratively perform data processing, learning, and decision-making directly in orbit, which is essential for latency-sensitive applications such as Earth observation, autonomous space operations, and intelligent network management.

A fundamental challenge in realizing SDCs lies in the training and deployment of large-scale machine learning (ML) models under severe communication constraints. State-of-the-art intelligent space applications typically rely on highly parameterized deep neural networks (DNNs) \cite{11112763}, whose centralized training requires frequent transmission of massive datasets or model parameters to ground stations (GSs) \cite{11036334}. However, the short visibility windows between fast-moving LEO satellites and GSs, coupled with limited downlink bandwidth, render centralized training architectures inefficient and often infeasible for low-latency in-orbit intelligence \cite{11263807}. These limitations motivate a shift toward \emph{distributed learning} paradigms that can fully exploit onboard computing resources while minimizing reliance on ground infrastructure.

Federated learning (FL) has emerged as a promising distributed optimization framework, allowing satellites to train local models and exchange only intermediate updates over inter-satellite links (ISLs) \cite{10353003}. Despite its appeal, deploying FL in LEO satellite networks remains challenging due to the stringent bandwidth and energy constraints of ISLs, as well as their intermittent connectivity and rapidly time-varying topology \cite{10608136,10646360,11132321}. In particular, the repeated exchange of high-dimensional model parameters or gradients can easily overwhelm the limited communication resources, leading to congestion, excessive latency, and vulnerability to single points of failure in centralized aggregation architectures \cite{10418548,10734153}. These characteristics make communication efficiency and robustness central concerns in distributed learning over satellite networks.

To alleviate the communication bottleneck, over-the-air computation (AirComp) has been proposed as a spectrum-efficient aggregation technique that exploits the superposition property of multiple-access channels \cite{10320326,10746330}. By allowing all nodes to transmit simultaneously, AirComp enables aggregation latency that is independent of the number of participating devices, making it particularly attractive for large-scale distributed learning \cite{6557530,9095231}. While AirComp has demonstrated significant gains in terrestrial wireless networks, its extension to satellite systems which rely on free space optical (FSO) inter-satellite links faces fundamental obstacles. Conventional AirComp relies on coherent signal superposition, which requires stringent phase alignment and accurate channel state information (CSI) to ensure constructive aggregation \cite{10315036}. In the space environment, however, maintaining such coherence is exceedingly difficult due to satellite jitter, pointing errors, and severe Doppler-induced phase variations caused by high relative velocities \cite{9852737,11260860}. These impairments can easily turn constructive superposition into destructive interference, severely degrading aggregation accuracy and destabilizing the learning process \cite{10302307}.

Motivated by the inherent fragility of coherent schemes, recent research has turned toward {non-coherent AirComp} architectures that eliminate the need for phase synchronization and instantaneous CSI. Energy-based aggregation schemes, which encode local updates into orthogonal signaling dimensions and rely on non-coherent energy detection at the receiver, have shown improved robustness in dynamic environments \cite{9272666,9771881,10008587}. In parallel, sign-based gradient methods such as signSGD have attracted considerable interest in distributed optimization due to their extreme communication efficiency and provable convergence properties under noisy updates \cite{bernstein2018signsgd,bernstein2019signsgd}. Nevertheless, combining non-coherent AirComp with sign-based aggregation introduces a new challenge: {aggregation bias} induced by heterogeneous channel conditions. In energy-domain majority-vote aggregation, satellites experiencing stronger channels can disproportionately influence the aggregated result, leading to biased gradient estimates and degraded learning performance. Existing power control strategies developed for coherent AirComp in terrestrial networks \cite{9844173,10078151} are not directly applicable, as they typically rely on accurate CSI and do not account for the statistical nature of non-coherent majority-vote detection over time-varying FSO channels.

In this paper, we address these challenges by proposing {OptiVote}, a robust and communication-efficient non-coherent AirComp framework tailored for distributed learning over FSO-based LEO satellite networks. OptiVote integrates signSGD with a majority-vote aggregation mechanism implemented via pulse-position modulation (PPM), enabling phase-agnostic aggregation through non-coherent energy accumulation. To further mitigate channel-induced aggregation bias without incurring additional signaling overhead, we develop an importance-aware, CSI-free dynamic power control scheme that balances the received energies across satellites based on their statistical contributions. Through rigorous theoretical analysis, we characterize the aggregation error probability under statistical FSO fading and establish convergence guarantees for non-convex optimization.  Extensive experiments demonstrate that OptiVote consistently outperforms representative baselines in both communication efficiency and learning accuracy. Our main contributions are summarized as follows:
\begin{itemize}
    \item 
    We design a non-coherent AirComp architecture that integrates signSGD, MV aggregation, and PPM-based signaling, enabling simultaneous uplink aggregation via energy accumulation without phase synchronization.

    \item 
    We develop an importance-aware dynamic power control mechanism that requires no instantaneous CSI or extra signaling, and effectively balances the received energies to reduce bias in energy-domain majority voting over heterogeneous FSO links.

    \item
    We characterize the aggregate decision error probability under statistical FSO fading and establish convergence guarantees for non-convex objectives under non-coherent MV aggregation with communication impairments.

    \item
    We conduct comprehensive experiments and show that OptiVote achieves improved communication efficiency and learning accuracy compared with representative baselines, supporting scalable and resilient in-orbit intelligence in communication-constrained satellite networks.
\end{itemize}
The remainder of this paper is organized as follows. Section~II presents the distributed FL formulation and the FSO communication model tailored for LEO satellite networks. Section~III details the transmitter/receiver design of OptiVote. Section~IV provides theoretical analysis, including aggregation error characterization and convergence results. Section~V reports numerical evaluations, followed by concluding remarks in Section~VI.

\section{Minority follows Majority: AirComp-based Distributed Strategy} 

\subsection{Distributed Federated Learning Model}

We consider an FL system comprising a single aggregating satellite (AS) that acts as the edge server, coordinating the learning process across $\emph{M}$ LEO satellites (space nodes), which act as edge devices.

In this system, at communication round $\emph{n}$, each space node $\emph{m}\in{[M]}$ computes the local gradient $\emph{g}_{m}^{(n)}$ and sends its sign to the AS. The AS performs aggregation and sends the MVs back to all space nodes. Finally, all the space nodes update the local models by utilizing the the MVs. 

Thus, at the $\emph{m}$-th space node, the learning objective is to solve a local optimization based on its local dataset $\mathcal{D}_{m}$:
\begin{equation}
\min _{\mathbf{w} \in \mathbb{R}^{q}} F_{m}(\mathbf{w})=\min _{\mathbf{w} \in \mathbb{R}^{q}} \frac{1}{|D_{m}|} \sum_{\forall(\mathbf{x}_{\ell}, y_{\ell}) \in D} f(\mathbf{w}; \mathbf{x}_{\ell}, y_{\ell}).    
\end{equation}
where $\mathcal{D}_{m}$ is denoted as the local data containing labeled data samples at $\emph{m}$-th space node as $\left\{\left(\mathbf{x}_{\ell}, y_{\ell}\right)\right\} \in \mathcal{D}_{m}$ for $m=1, \ldots, M$,
and $\mathbf{x}_{\ell}$ and $y_{\ell}$ are ${\ell}$-th data sample and its associated label, respectively. $\mathcal{D}=\bigcup_{m=1}^{M}\left\{\mathcal{D}_{m}\right\}$ is the global dataset set and $f(\mathbf{w}, \mathbf{x}, y)$ is the sample-wise loss function indicating the prediction error for example $(\mathbf{x}, y)$ with the FL model
parameters $\mathbf{w}=\left[w_{1}, \ldots, w_{q}\right]^{\mathrm{T}} \in \mathbb{R}^{q}$ , and $\emph{q}$ is the number of model parameters.

The ultimate goal of this centralized FL system is to find the optimal model parameter $\mathbf{w}^{*}$ that minimizes $F(\mathbf{w})$over all distributed datasets, i.e.,
\begin{equation}\label{eq2}
\mathbf{w}^{*} = \min _{\mathbf{w} \in \mathbb{R}^{q}} F(\mathbf{w})=\min _{\mathbf{w} \in \mathbb{R}^{q}} \frac{1}{|D|} \sum_{\forall(\mathbf{x}, y) \in D} f(\mathbf{w}; \mathbf{x}, y),
\end{equation}
where $F(\mathbf{w})$ is the global loss function. Our focus is on the uplink gradient aggregation process using non-coherent AirComp, as shown in Fig. \ref{fig1}. 

\begin{figure*}[t]
\centering
\includegraphics[width=1\linewidth]{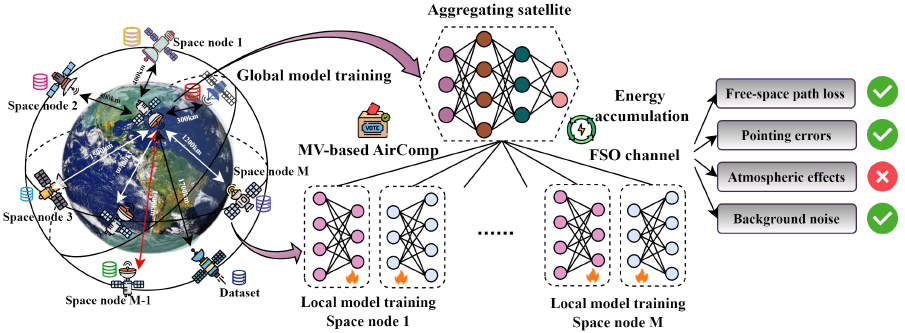}
\caption{Transceiver design on distributed Learning via non-coherent over-the-air computation based on importance-aware majority vote strategy.}
\label{fig1}
\end{figure*}

Specifically, all space nodes and AS periodically transmit model parameters by uplink and downlink communication, and space nodes involved in the learning process share the same FSO channel medium. Without loss of generality, at the $\emph{n}$-th communication round, the following processing will be:
\begin{itemize}
\item \textbf{Preparation phase:} The AS selects $\emph{M}$ active space nodes for model training procedure; 
\item \textbf{DL broadcast:} The AS pre-broadcasts global models $\mathbf{w}^{(n)}$ and global gradient $\mathbf{g}_{\ell}^{(n)}$ for this round to selected space nodes, which work together to train FL model.
\item \textbf{Local model update:} Each space node $\emph{m}$ exploits a mini-batch stochastic gradient descent (SGD) method to calculate its local gradient $\tilde{\mathbf{g}}_{m}^{(n)} \triangleq[\tilde{\emph{g}}_{m, 1}^{(n)}, \ldots, \tilde{\emph{g}}_{m, q}^{(n)}]^{\mathrm{T}}$ with respect to the selected data batch and the current received global model $\mathbf{w}^{(n)}$ as
\begin{equation}
\tilde{\mathbf{g}}_{m}^{(n)}=\nabla F_{m}\left(\mathbf{w}^{(n)}\right)=\frac{1}{d_{\mathrm{b}}} \sum_{\forall\left(\mathbf{x}_{\ell}, y_{\ell}\right) \in \tilde{\mathcal{D}}_{m}} \nabla f\left(\mathbf{w}^{(n)}, \mathbf{x}_{\ell}, y_{\ell}\right),
\end{equation}
where $\nabla$ represents the gradient operator and $\tilde{\mathcal{D}}_{m} \subset\mathcal{D}_{m}$ is selected data batch from local data set and $d_{\mathrm{b}}=|\tilde{\mathcal{D}}_{m}|$ as the batch size. The initial local model is $\mathbf{w}^{(n,0)}_{m} = \mathbf{w}^{n}$, and the local model update is represented as
\begin{equation}
\mathbf{w}^{(n+1)}=\mathbf{w}^{(n)}-\eta \mathbf{g}^{(n)},
\end{equation}
where $\mathbf{g}^{(n)}$ represent the gradient of the global loss $F(\mathbf{w}^{(n)})$ with respect to $\mathbf{w}^{(n)}$. 
\item \textbf{UL aggregation process:} The selected space nodes upload the local stochastic gradients to AS over the FSO links, denoted as $\tilde{\mathbf{g}}_{m}^{(n)}$ for $m=1, \ldots, M$. Hence, the $\emph{i}$-th element of the global aggregated gradient at the AS would be computed as follows: 
\begin{equation}
\mathbf{g}_{\ell}^{(n)}=\sum_{m=1}^{M}{\tilde{\mathbf{g}}_{m, i}^{(n)}}.
\end{equation}
\item \textbf{Global model update:} The AS receives the aggregated gradients by uploaded from the space nodes and then conducts a global model update.
\end{itemize}

This iterative procedure repeats consecutively until a predetermined convergence criterion is achieved. 

\subsection{SignSGD With Majority Vote}

Under the background of FL, we combine SignSGD \cite{bernstein2018signsgd} with a majority vote strategy to solve the optimization objective in Eq.(\ref{eq2}). This approach drastically reduces the uplink communication payload by transmitting only the specific sign of the stochastic gradients.

Thus, the renewed aggregation process is shown as follows:
\begin{itemize}
\item \textbf{One-bit quantization:} All selected EDs upload the signs of their local stochastic gradients to the AS over the FSO links, which are denoted as $\bar{\mathbf{g}}_{m}^{(n)}$ for $m=1, \ldots, M$. The $\emph{i}$-th element of the sign vector from $\emph{m}$-th space nodes is given by:
\begin{equation}  
\bar{\mathbf{g}}_{m,i}^{(n)} = \operatorname{sign}(\tilde{\mathbf{g}}_{m, i}^{(n)}).
\end{equation} 
\item \textbf{Majority vote estimation:} The AS enforces the parameter MV as the estimate of the $\emph{i}$-th global gradient element. In an ideal FSO channel, the MV is computed as:
\begin{equation}
v_{i}^{(n)} = \operatorname{sign}\left(\sum_{m=1}^{M} \bar{\mathbf{g}}_{m, i}^{(n)}\right).
\end{equation}
This process is what is commonly known as majority vote, where each space node votes with the sign of its true gradient.
\item \textbf{Global model update:} Afterwards, the AS pushes the MV vector $\mathbf{v}^{(n)}=[v_{1}^{(n)}, \ldots, v_{q}^{(n)}]^{\mathrm{T}}$ back to the space nodes. The models at the space nodes are updated based on the MV direction: 
\begin{equation}
\mathbf{w}^{(n+1)}=\mathbf{w}^{(n)}-\eta \mathbf{v}^{(n)}.
\end{equation}
\end{itemize}  

\begin{remark} 
\emph{Space Nodes transmit their parameters over FSO links that are inherently unreliable due to channel fading, pointing errors, and noise. These factors introduce transmission errors that, when aggregated, degrade FL performance. Also, the uplink process from the space nodes to the AS is more constrained than the downlink in terms of both bandwidth and network throughput, thus necessitating a focus on uplink communication. Meanwhile, we assume that $\mathbf{w}^{(n)}$ transmitted via the downlink channel can be received by all moving space nodes in an error-free manner. The key challenge is to calculate an accurate estimate on MV vector $\mathbf{v}^{(n)}$.}
\end{remark}

\subsection{Free-space Optics Link Integrated Aircomp}
We concentrate on the uplink communication process for non-coherent FSO AirComp. The aggregation of quantized signs $\bar{\mathbf{g}}_{m,i}^{(n)}$ is accomplished by exploiting the linear superposition property of optical intensity in the FSO channel.

Specifically, each space node transmits its one-bit gradient sign by utilizing PPM. And PPM encodes the sign ($\pm 1$) by activating one of $\emph{A}$ orthogonal time slots ($\tau \in \{0, 1, \ldots, A-1\}$) on the FSO link. The FSO channel coefficient $\emph{I}_{m}^{(n)}$ from space node $\emph{m}$ to the AS is modeled by the combined effect of geometric path loss, atmospheric turbulence, and pointing errors (Jitter), leading to optical intensity fluctuations. We assume a block fading scenario where $\emph{I}_{m}^{(n)}$ remains constant within a communication round $\emph{n}$.

Thus, at the $\emph{n}$-th communication round, the instantaneous received optical intensity $r_{\tau}^{(n)}$ at the AS in the $\tau$-th orthogonal PPM slot is given by the superposition of intensities from all $\emph{M}$ active space nodes:
\begin{equation}
\emph{r}_{\tau}^{(n)} = \sum_{m=1}^{M} C_{R}P_{m}^{(n)} I_{m}^{(n)}t_{m, \tau}^{(n)} + \epsilon_{\tau}^{(n)} \text{,} \label{eq:fso_intensity_final_tau}
\end{equation}
where $C_{R} \triangleq \eta_{R} G_{R}=1$ is the constant receiver gain, encompassing $\eta_{R}$ (receiver optical efficiency) and $G_{R}$ (receiver telescope gain) and $P_{m}^{(n)}$ is the instantaneous transmit optical power of $\emph{m}$-th space node. Moreover, $\emph{I}_{m}^{(n)}$ is the instantaneous intensity attenuation factor of the FSO channel. $t_{m, \tau}^{(n)} \in \{0, 1\}$ is the PPM signal transmitted by $\emph{m}$-th space node in slot $\tau$ ( $t_{m, \tau}^{(n)}$ equals to 1 if the sign corresponds to slot $\tau$, and 0 otherwise). $\epsilon_{\tau}^{(n)}\sim \mathcal{N}\left(0, \sigma^{2}\right)$ is the additive noise, typically modeled as a Gaussian process.

Also, $P_{m}^{(n)}$ is the instantaneous transmit power. The transmission of each space node is subject to a long-term average power constraint:
\begin{equation}
\mathbb{E}\left[P_{m}^{(n)}\right] \leq P_{\text{avg}}, \forall m \text{,} \label{eq:power_constraint_final_tau}
\end{equation}
where $P_{\text{avg}} > 0$ is the maximum average transmit power.

With this non-coherent FSO scheme, the channel attenuation factor $I_{m}^{(n)}$ directly weights the contribution of each space node's vote in the received intensity. This causes the MV mechanism to suffer from aggregation weight imbalance (clustering bias), as satellites with favorable channel conditions dominate the aggregation. Given the impracticality of relying on instantaneous CSI due to high orbital velocity, it is necessary to design an accurate scheme to offset this bias, thereby enhancing the overall system efficiency.

\begin{figure*}[t]
\centering
\includegraphics[width=1\linewidth]{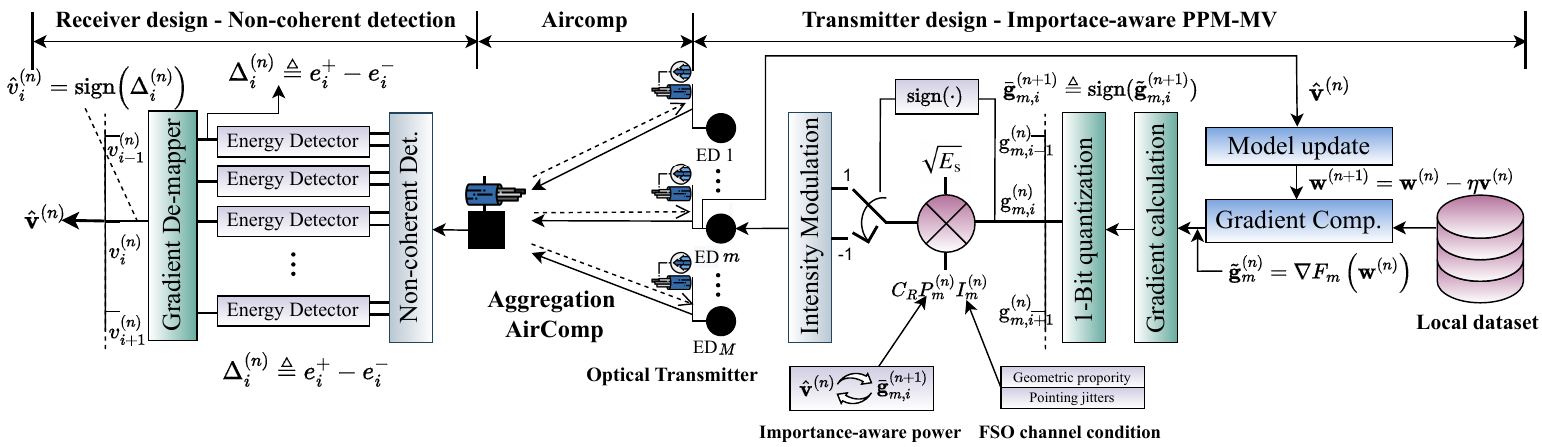}
\caption{Transceiver design on distributed Learning via non-coherent over-the-air computation based on importance-aware majority vote strategy.}
\label{fig2}
\end{figure*}


\section{OptiVote: Phase-Less Signal Aggregation and Importance-Aware Power Design} 
In this section, we develop a non-coherent FSO AirComp framework shown in Fig. \ref{fig2}, which enables reliable MV-based model aggregation without requiring explicit CSI or phase synchronization. 

\subsection{Transmitter Design - Importance-Aware PPM-MV}
To avoid stringent phase alignment demands in FSO links, we design an importance-aware transmitter architecture predicated on PPM-based MV (PPM-MV) with AirComp strategy. 

Specifically, our scheme allocates $A$ orthogonal time slots, where each space node performs a low-complexity operation to transmit one-bit gradient signs using binary PPM over FSO links. 
Let $f$ be a bijective mapping that assigns each gradient coordinate $i\in\{1,2,\ldots,q\}$ to a distinct ordered pair of orthogonal time slots $(\tau_i^+,\tau_i^-)$ within a frame, where $\tau_i^+,\tau_i^-\in\{1,\ldots,A\}$ and $\tau_i^+\neq\tau_i^-$. 
At round $n$, node $m$ applies one-bit sign quantization
$\overline{\mathbf{g}}_{m,i}^{(n)}\triangleq \mathrm{sign}(\tilde{\mathbf{g}}_{m,i}^{(n)})\in\{-1,+1\}$.
Based on $\overline{\mathbf{g}}_{m,i}^{(n)}$, the transmitted optical symbol at slot $\tau_i^+$ is
\begin{equation}\label{eq11}
t_{m,\tau_i}^{(n)}=
\begin{cases}
\sqrt{E_s}, & \tau=\tau_i^+,~\overline{\mathbf{g}}_{m,i}^{(n)}=+1,\\
0, & \tau=\tau_i^-,~\overline{\mathbf{g}}_{m,i}^{(n)}=-1,
\end{cases}
\end{equation}
and the transmitted optical symbol at slot $\tau_i^-$ is
\begin{equation}\label{eq12}
t_{m,\tau_i}^{(n)}=
\begin{cases}
0, & \tau=\tau_i^-,~\overline{\mathbf{g}}_{m,i}^{(n)}=-1,\\
\sqrt{E_s}, & \tau=\tau_i^+,~\overline{\mathbf{g}}_{m,i}^{(n)}=+1,
\end{cases}
\end{equation}
respectively, where $\sqrt{E_s}=1$ is a factor to normalize the symbol energy. As opposed to the OBDA scheme in \cite{9272666}, which relies on complex amplitude modulation and thus necessitates strict phase synchronization, our scheme exploits Eq.(\ref{eq11}) and Eq.(\ref{eq12}) to indicate two orthogonal time slots for uploading the signs of the local stochastic gradients. To transmit the encoded gradients, we use the PPM-MV modulation scheme, which is discussed as follows:

\begin{itemize}
    \item \textit{General PPM Configuration}: The effective time slots serve as the resource for encoding gradients, with transmitted signals consisting of optical pulses. Consequently, the AirComp scheme designates two distinct positions within a frame for voting. Specifically, this design aims to motivate the receiver detect the MVs via non-coherent energy accumulation, thereby bypassing the requirement for instantaneous CSI.
    \item \textit{Adjacency-Based Mapping (PPM-MV)}: As a robust configuration on the mapping function $f$, we choose the two voting slots to be adjacent within a frame, i.e., $\tau_i^-=\tau_i^+ + 1$ (or more generally $|\tau_i^+-\tau_i^-|T_s \ll T_c$) for all $i$. 
    Since the two neighboring slots fall within the same coherence interval, they experience highly correlated intensity fading in FSO links (e.g., due to turbulence and pointing jitter), i.e.,
    $I_{m,\tau_i^+}^{(n)} \approx I_{m,\tau_i^-}^{(n)}$. 
    This sign-independent fading symmetry facilitates reliable MV recovery via differential energy detection by comparing the accumulated energies across the two slots. We refer to this specific adjacency-based design as PPM-based MV.
\end{itemize}

However, the intensity-domain superposition in FSO AirComp introduces a critical challenge, i.e., aggregation bias. 
In distributed architectures, factors such as data availability and training gains result in an inherently distinct contribution to the global model across satellites. Ideally, aggregation decisions should reflect this perceived importance rather than being dominated by channel-stochasticity, which leads to unfair aggregation and degraded convergence.





\textbf{Importance-Aware adaptive Power Allocation:} During the distributed learning, satellites may contribute unequally to the global update due to heterogeneous data and dynamics training, and non-coherent power accumulation may further introduce channel-induced energy dominance. 

To regulate the effective contribution of each satellite and mitigate aggregation bias without instantaneous CSI, we design an importance-aware power control strategy driven by a unified importance weight that fuses data significance, local training gain, and directional consistency.

Specifically, after the AS obtains the MV decision $\hat{\mathbf v}^{(n)}=[\hat v_1^{(n)},\ldots,\hat v_q^{(n)}]$ via differential energy detection and broadcasts it, satellite $m$ evaluates an importance-consistency score based on the agreement between its local sign vector and the broadcast MV:
\begin{equation}
a_m^{(n)} \triangleq \frac{1}{q}\sum_{i=1}^{q}\mathbb{I}\!\left[\overline {\mathbf{g}}_{m,i}^{(n)}=\hat v_i^{(n-1)}\right]\in[0,1],
\end{equation}
where $\overline{\mathbf{g}}_{m,i}^{(n)}=\mathrm{sign}(\tilde{\mathbf{g}}_{m,i}^{(n)})$ denotes the one-bit local gradient sign.
Thus, the transmit optical power is updated via the projected recursion
\begin{equation}
P_m^{(n+1)}=\Pi_{[P_{\min},\,P_{\max}]}\!\Big(
P_m^{(n)}+\rho\big(a_m^{(n)}-\bar{a}^{(n)}\big)
\Big),
\end{equation}
where,
\begin{equation}
\bar{a}^{(n)} \triangleq \frac{1}{M}\sum_{m=1}^{M}a_m^{(n)} .
\end{equation}
where $\rho>0$ is a stepsize and $\Pi_{[P_{\min},P_{\max}]}(\cdot)$ denotes projection onto the feasible power interval.

Since the MV-based global direction can be regarded as error-free (or near error-free), we assign transmit powers based on the agreement between each space node's local sign and the global decision rather than using identical powers. This CSI-free adjustment makes the aggregated update increasingly consistent with the desired convergence direction.

\subsection{Receiver Design - Non-coherent Energy Detection}
With PPM-MV procedure, the AS performs CSI-free uplink aggregation by differential energy detection over the adjacent slot pair $(\tau_i^+,\tau_i^-)$ for each coordinate $i$. Specifically, the accumulated received energies at the AS are
\begin{equation}\label{eq16}
e_{i}^{(n),+} \triangleq \emph{r}_{\tau_i^+}^{(n)}
= \sum_{m:\,\bar{\mathbf{g}}_{m,i}^{(n)}=+1} C_R P_m^{(n)} I_m^{(n)} \sqrt{E_s} \;+\; \epsilon_{\tau_i^+}^{(n)},
\end{equation}
and
\begin{equation}\label{eq17}
e_{i}^{(n),-} \triangleq \emph{r}_{\tau_i^-}^{(n)}
= \sum_{m:\,\bar{\mathbf{g}}_{m,i}^{(n)}=-1} C_R P_m^{(n)} I_m^{(n)} \sqrt{E_s} \;+\; \epsilon_{\tau_i^-}^{(n)},
\end{equation}
respectively. The receiver at the AS detects the MV for the $i$-th gradient with an energy detector as
\begin{equation}
\Delta_i^{(n)} \triangleq e_{i}^{(n),+} - e_{i,-}^{(n),+},
\end{equation}
and recovers the MV decision by
\begin{equation}\label{eq19}
\hat v_i^{(n)} = \mathrm{sign}\!\left(\Delta_i^{(n)}\right),
\end{equation}
where $\hat{\mathbf v}^{(n)}=[\hat v_1^{(n)},\ldots,\hat v_q^{(n)}]$. Note that we do not utilize any strategy to resolve interference between space nodes in Eq.(\ref{eq16}) and Eq.(\ref{eq17}), as we are not concerned with the sign of the local gradient. Instead, we leverage this interference for aggregation by comparing the energy values on paired PPM time slots $(\tau_i^+,\tau_i^-)$ to detect MVs in Eq. (\ref{eq19}).

Then, the AS broadcasts $\hat{\mathbf v}^{(n)}$ to all satellites, and the global model is updated along the MV direction as
\begin{equation}
\mathbf w^{(n+1)} = \mathbf w^{(n)} - \eta\, \hat{\mathbf v}^{(n)}.
\end{equation}

In practice, MVs can be transmitted by conventional methods, for which the detection is all with some errors. In contrast, the reception of MV vectors by Space nodes is assumed to be perfect in the scenarios of \cite{jin2021stochasticsign}. Thus, we assume that the MVs can be transmitted to the Space node perfectly, due to the high transmit power available at the AS and the use of the whole downlink bandwidth for broadcasting.

\section{Error Probability Analysis and
Convergence Rate Performance}

In this section, we theoretically characterize OptiVote by quantifying the error probability of MV detection under fading channels and establish its convergence performance.

\subsection{Reliable Assumptions on Theoretical Analysis}
To permit further convergence and error probability analysis, we specify several standard assumptions \cite{bernstein2018signsgd}, \cite{9272666}. Moreover, to extend the theoretical framework beyond convex loss functions to neural networks, we introduce lower bound requirements \cite{allenzhu2018natasha}.

\begin{assumption}[Lower-bounded Objective]
\label{Assumption1}
For all parameter vectors $\mathbf{w}$, the objective function $F(\mathbf{w})$ is bounded below by a finite constant $F^*$, i.e., $F(\mathbf{w}) \ge F^*$, $\forall \mathbf{w}$.  
\end{assumption}

\begin{assumption}[Smoothness]
Let $\mathbf{g}\triangleq\nabla F(\mathbf{w})$ denote the gradient of the loss function $F(\mathbf{w})$ evaluated at $\mathbf{w}$. For all $\mathbf{w}$ and $\mathbf{w'}$, the expression is given by
\begin{equation}
F(\mathbf w') \le F(\mathbf w) + \nabla F(\mathbf w)^{\mathsf T}(\mathbf w'-\mathbf w)
+ \frac{1}{2}\sum_{i=1}^{q} L_{i}\left(\emph{w}_{i}^{\prime}-\emph{w}_{i}\right)^{2},
\end{equation}
\end{assumption}
where we can assume that there exists a vector of non-negative constants $\mathbf{L}=\left[L_{1}, \ldots, L_{q}\right]^{\mathrm{T}}$.

\begin{assumption}[Bounded Variance]
The stochastic gradient estimates $\{\tilde{\mathbf{g}}_{m}=[\tilde{\text{g}}_{m, 1}, \ldots, \tilde{\text{g}}_{m, q}]^{\mathrm{T}}=\nabla F_{m}(\mathbf{w})\}$ are independent and unbiased estimates of  $\mathbf{g}\!=\![\text{g}_{1}, \ldots,\text{g}_{q}]^{\mathrm{T}}\!=\!\nabla F(\mathbf{w})$ with a coordinate bounded variance, i.e.,
\begin{equation}
\mathbb{E}\left[\tilde{\mathbf{g}}_{\emph{m}}\right]=\mathbf{g}, \forall \emph{m},
\end{equation}
\begin{equation}
 \mathbb{E}\left[\left(\tilde{\text{g}}_{\emph{m}, \emph{i}}-\text{g}_{\emph{i}}\right)^{2}\right] \leq \sigma_{\emph{i}}^{2} / d_{\mathrm{b}}, \forall \emph{m}, \emph{i},
\end{equation}
\end{assumption}
where $\boldsymbol{\sigma}=\left[\sigma_{1}, \ldots, \sigma_{q}\right]^{\mathrm{T}}$ is a non-negative constant vector, and where $\tilde{\text{g}}_{k, i}$ and $\text{g}_{i}$ denote the $\emph{i}$-th element of $\tilde{\mathbf{g}}_{k}$ and $\mathbf{g}$.

Apart from these above analysis, another significant assumption is that the data-stochasticity induced gradient noise. Meanwhile, this assumption causes the discrepancy between $\tilde{\mathbf{g}}_{k}$ and $\mathbf{g}$, which is unimodal described as follow.

\begin{assumption}[Unimodal, Symmetric Gradient Noise]
For any given $\mathbf{w}$, each elements of
$\tilde{\mathbf{g}}_{\emph{m}}$, $\forall \emph{m}$, has a unimodal
distribution that is also symmetric around its mean.
\end{assumption}

\subsection{Received Signal Power of MV}

The proposed OptiVote scheme leads to a fundamentally different aggregation strategy compared to coherent AirComp, as it determines the correct MVs in Eq. (\ref{eq19}) probabilistically by comparing the accumulated energies $e_{i}^{+}$ and $e_{i}^{-}$. 

To elaborate on this, we consider the specific characteristics of inter-satellite links, i.e., the optical signal propagates in a vacuum environment, effectively eliminating atmospheric turbulence.   Consequently, the channel impairments are primarily governed by deterministic geometric path loss and stochastic pointing errors caused by platform vibrations.  
We assume these stochastic impairments and the additive noise are independent across different space nodes.

Let $M_{\emph{i}}^{+}$ and $M_{\emph{i}}^{-}= M-M_{\emph{i}}^{+}$ denote the number of space nodes that vote for $+1$ and $-1$ for the $\emph{i}$-th gradient element, respectively. Based on the proposed system, we derive the means of accumulated energies, denoted as $\mu_{\emph{i}}^{+}$ and $\mu_{\emph{i}}^{-}$, in the following lemma.

\begin{lemma}[Energy Detection]
For a given vote allocation $M_{i}^{+}$ and $M_{i}^{-}$, $\mu_{i}^{+} \triangleq \mathbb{E}[e_{i}^{+}]$ and $\mu_{i}^{-} \triangleq \mathbb{E}[e_{i}^{-}]$, are calculated as:
\begin{equation}
\begin{aligned}
\mu_{i}^{+} \triangleq \mathbb{E}[e_{i}^{+}] &= M_{i}^{+} \vartheta + \sigma_n^2, \label{eq:mu_plus}
\end{aligned}
\end{equation}
and,
\begin{equation}
\begin{aligned}
\mu_{i}^{-} \triangleq \mathbb{E}[e_{i}^{-}] &= M_{i}^{-} \vartheta + \sigma_n^2, \label{eq:mu_minus}
\end{aligned}
\end{equation}
respectively. Here, $\sigma_{n}^{2}$ is the noise variance, and $\vartheta$ represents the received signal energy per satellite. Considering the proposed adaptive scheme and the FSO channel impairments, $\vartheta$ is defined as:
\begin{equation}
\label{eq:theta_def}
\vartheta = C_R \sqrt{E_s} \cdot \underbrace{P_{\text{avg}}}_{\text{Power Control}} \cdot \underbrace{\lambda}_{\text{Channel Interference}},
\end{equation}
where:
\begin{itemize}
    \item $P_{\text{avg}} \triangleq \mathbb{E}[P_{m}]$ is the average transmit optical power governed by the proposed scheme. Due to the projection operation $\Pi_{[P_{\min}, P_{\max}]}(\cdot)$, $P_{m}^{(n)}$ evolves as a bounded stochastic process. Consequently, its expectation is well-defined and strictly constrained within $[P_{\min}, P_{\max}]$, typically stabilizing around the target average budget $P_{\text{avg}}$.
    
    \item $\lambda \in (0, 1]$ denotes the normalized FSO channel efficiency factor. Based on the independence between spatial topology and pointing jitters, $\lambda$ is derived as the product of geometric efficiency and pointing efficiency:
    \begin{equation}\label{eq27}
    \lambda = \underbrace{\left( \frac{3 C_{\text{FSPL}} (D_{\max} - D_{\min})}{D_{\max}^3 - D_{\min}^3} \right)}_{\text{Geometric Efficiency}} \cdot \underbrace{\left( A_0 \frac{\xi^2}{\xi^2 + 1} \right)}_{\text{Pointing Efficiency}}.
    \end{equation}
    The first term accounts for the ensemble average path loss over the 3D spherical shell distribution of satellites, while the second term quantifies the energy reduction due to stochastic pointing errors modeled by the Beckmann distribution parameter $\xi$.
\end{itemize}
\end{lemma}

\begin{proof}
    Please refer to Appendix \ref{app:l1}.
\end{proof}

\subsection{Bit Error Probability Analysis}

To facilitate the error probability analysis, we adopt the following standard assumptions regarding gradient statistics and signal distributions.
\begin{assumption}[Independent, identical, and unbiased gradients]
The local stochastic gradient estimates are independent
and unbiased, i.e., $\mathbb{E}_{\tilde{\mathcal{D}}_{m}}[\tilde{\text{g}}_{m, i}^{(t)}]=\text{g}_{i}^{(n)}, \forall m, i$.
\end{assumption}

\begin{assumption}[Statistical Probability Distribution on Energy Detection] 
Given the vote counts $\emph{M}_{\emph{i}}^{\textbf{\texttt{+}}}$ and $\emph{M}_{\emph{i}}^{\textbf{\texttt{-}}}$, the accumulated decision variables $\emph{e}_{\emph{i}}^{\textbf{\texttt{+}}}$ and $\emph{e}_{\emph{i}}^{\textbf{\texttt{-}}}$ are modeled as independent random variables following an exponential distribution with means $\mu_{\emph{i}}^{\textbf{\texttt{+}}}$ and $\mu_{\emph{i}}^{\textbf{\texttt{-}}}$, respectively.
\end{assumption}

\begin{remark}
\emph{From a physical perspective, the aggregated energy is characterized by a Gamma distribution. However, an exact error probability analysis under this model leads to mathematically intractable expressions involving hyper-geometric functions. To facilitate a tractable closed-form derivation, we adopt the exponential distribution model. Since the exponential distribution corresponds to a Gamma distribution with a shape parameter $G=1$, it exhibits a heavier tail compared to the actual aggregated signals ($G>1$). This assumption effectively models the worst-case fading scenario, enabling the derived error probability a strict and conservative upper bound on the error probability.}
\end{remark}

Let $\emph{P}_\emph{i}^{\text{err}}$ denote the probability on misidentifying the sign of $\emph{i}$-th global gradient element $\text{g}_{\emph{i}}^{(\emph{n})}$. Without loss of generality, assume the true global sign is positive, i.e., $\text{sign}(\text{g}_{\emph{i}}^{(\emph{n})}) \texttt{=} \texttt{+}1$. This error probability is defined as:
\begin{equation}\label{eq39}
\emph{P}_\emph{i}^{\text{err}} \triangleq \mathbb{P}\left(\text{sign}\left(\Delta_{\emph{i}}^{(\emph{n})}\right) \neq \text{sign}\left(\text{g}_{\emph{i}}^{(\emph{n})}\right)\right) = \mathbb{P}\left(\emph{e}_\emph{i}^\textbf{\texttt{+}} < \emph{e}_\emph{i}^\textbf{\texttt{-}}\right).
\end{equation}

The aggregation correctness depends on the number of space nodes that correctly vote for $\texttt{+}1$. Thus, let $\emph{Z}$, be a random variable for counting the number of space nodes with the correct decision i.e., $\text{sign}(\text{g}_{\emph{i}}^{(\emph{n})}) \texttt{=} \texttt{+}1$. The random variable $\emph{Z}$ can then be modeled as the sum of $\emph{M}$ independent Bernoulli trials, i.e., a binomial variable with the success and failure probabilities given by
\begin{align}
&\emph{p}_\emph{i} \triangleq \mathbb{P}\left[\text{sign}\left(\tilde{\text{g}}_{\emph{m,i}}^{(\emph{n})}\right)=\text{sign}\left(\text{g}_{\emph{i}}^{(\emph{n})}\right)\right],\\
&\emph{q}_\emph{i} \triangleq \mathbb{P}\left[\text{sign}\left(\tilde{\text{g}}_{\emph{m,i}}^{(\emph{n})}\right) \neq \text{sign}\left(\text{g}_{\emph{i}}^{(\emph{n})}\right)\right],
\end{align}
respectively, for all $\emph{m}$. where $\emph{p}_\emph{i} + \emph{q}_\emph{i}\texttt{=}1$, which are intuitively determined by the randomness of the data. For a more accurate formalization, $\emph{P}_\emph{i}^{\text{err}}$ for MV can be obtained as follows:

\begin{lemma}[Error Probability Upper Bound]
Let $\xi \triangleq \vartheta / \sigma_n^2$ denote the effective Signal-to-Noise Ratio (SNR) per satellite. Under the proposed OptiVote scheme, the error probability for the $i$-th gradient element is bounded by:
\begin{equation}
\label{eq:error_bound}
P_{i}^{\mathrm{err}} \leq \underbrace{\frac{M q_i}{M + 2/\xi}}_{\text{Data Noise Impact}} + \underbrace{\frac{1/\xi}{M + 2/\xi}}_{\text{Channel Noise Impact}}.
\end{equation}
Furthermore, incorporating the gradient noise bound for unimodal symmetric distributions, we have:
\begin{equation}
\emph{P}_{\emph{i}}^{\mathrm{err}} \leq \frac{\emph{M}({\sqrt{2} \alpha_{\emph{i}}})/({3\left|\text{g}_{\emph{i}}^{(\emph{n})}\right| \sqrt{\emph{d}_{\emph{b}}}})}{\emph{M}\texttt{+} 2/\xi} + \frac{1/\xi}{\emph{M}\texttt{+} 2/\xi}.
\end{equation}
\end{lemma}
where The coefficient $\emph{d}_{\emph{b}}$ is
that each local gradient estimate is computed over a mini batch of size $\emph{d}$ Also, the resultant gradient variance reduces from $\alpha^2$ to $\alpha^2/\emph{d}_{\emph{b}}$. 
\begin{proof}
    Please refer to Appendix \ref{app:l2}.
\end{proof}

\begin{lemma}[Failure Probability] 
Based on several previous assumptions in Section IV-A, the failure probability satisfies:
\begin{equation}\label{eq62}
\begin{aligned}
\emph{q}_{\emph{i}}  &=\mathbb{P}\left[\text{sign}\left(\tilde{\text{g}}_{\emph{m,i}}^{(\emph{n})}\right)\neq\text{sign}\left(\text{g}_{\emph{i}}^{(\emph{n})}\right)\right]\\
& \leq\left\{\begin{array}{ll}
\frac{2}{9} \frac{\alpha_{\emph{i}}^{2}}{\emph{d} _{\emph{b}}|\text{g}_{\emph{i}}^{(\emph{n})}|^{2}}, & \text { if } \frac{|\text{g}_{\emph{i}}^{(\emph{n})}|}{\alpha_{i} / \sqrt{\emph{d}_{\emph{b}}}}>\frac{2}{\sqrt{3}} \\
\frac{1}{2}-\frac{|\text{g}_{\emph{i}}^{(\emph{n})}|}{2 \sqrt{3} \alpha_{i} / \sqrt{\emph{d}_{\emph{b}}}}, & \text { otherwise, }
\end{array} \right.
\end{aligned}
\end{equation}
which is in all cases less than $1/2$. 
\label{lemma3}
\end{lemma}

\begin{proof}
    Please refer to Appendix \ref{app:l3}.
\end{proof}

Lemma \ref{lemma3} implies the following results:
\begin{corollary}[Legitimate Space Nodes]
For $\emph{q}_{\emph{i}} < \emph{p}_{\emph{i}}$, $\emph{M}_{i}^{+}$ must be larger than $\emph{M}/2$, meanwhile satisfy $\emph{P}_\emph{i}^{\text{err}} < 1/2$.
\end{corollary}

\subsection{Convergence Rate over Fading Channel}

The proposed OptiVote scheme incorporates probabilistic aggregation errors caused by FSO channel impairments, which may deviate from the global update on true MVs. Specifically, by incorporating specific FSO channel statistics, pointing errors, etc., the convergence rate is derived as follows:
\begin{theorem}[Convergence Rate]
Consider an FL system based on the proposed scheme, for the mini-batch size $\emph{d}_{\emph{b}}=N / \gamma$ and the learning rate $\eta=1 / \sqrt{\|\mathbf{L}\|_{1} \emph{d}_{\emph{b}}}$, the convergence rate in fading channel is given by
\begin{equation} 
\begin{split}
\mathbb{E}\left[\frac{1}{\emph{N}} \sum_{\emph{n}=0}^{\emph{N-1}}\left\|\mathbf{g}^{\emph{(n)}}\right\|_{1}\right] \leq \frac{1}{\sqrt{\emph{N}}}&\Big(\delta \sqrt{\|\mathbf{L}\|_{1}}\left(F(\mathbf{w}^{(0)})-\emph{F}^{*}+\frac{\gamma}{2}\right)\\&+\frac{2 \sqrt{2}}{3} \sqrt{\gamma}\|\boldsymbol{\sigma}\|_{1}\Big), 
\end{split}
\end{equation}
\end{theorem}
where $\gamma$ is a positive integer, $\delta\!=\!\left(1+\frac{2}{\xi \emph{M}}\right) \frac{1}{\sqrt{\gamma}}$ for $\xi \triangleq {\vartheta}/{\sigma_{\mathrm{n}}^{2}}$. $\vartheta$ characterizes the effective signal energy and evolves as a bounded variable, thereby guaranteeing a theoretical upper bound on the convergence performance.

\begin{proof}
    Please refer to Appendix \ref{app:T1}.
\end{proof}

\begin{remark}
\emph{Based on the derived convergence bound and the parameter definitions in Theorem 1, we can infer the following insights regarding the impact of system parameters:
\begin{itemize}
\item Impact of Scale and SNR: Regarding a larger effective SNR (i.e., a larger $\xi \propto \vartheta/\sigma_{n}^{2}$) and a massive number of participating satellites (i.e., a larger $M$), the convergence error term $\delta$ decreases explicitly. This indicates that expanding the constellation scale or improving the link budget directly accelerates the global model convergence.
\item Physical Channel Constraints: $\lambda$ related to geometric path loss and stochastic pointing errors act as attenuation factors on effective signal energy $\vartheta$. Specifically, a severe beam misalignment or a sparse spatial topology leads to a smaller $\lambda$, which lowers convergence rate by reducing the magnitude of the aggregated gradient signal. 
\item Role of Adaptive Power Control: The proposed adaptive scheme accounts for a more robust convergence performance compared to fixed power schemes. Since the transmit power is strictly constrained within $[P_{\min}, P_{\max}]$, $\vartheta$ evolves as a bounded variable, thereby guaranteeing that the convergence upper bound remains theoretically stable and finite. 
\end{itemize}
}
\end{remark}

Note that the proposed OptiVote scheme eliminates the stringent requirement for phase synchronization, thereby making the aggregation robust against phase jitter and mitigating the impact of amplitude fading effectively, thereby ensuring stable convergence behavior.

\section{Simulation Results} 

In this section,we conduct comprehensive experiments to compare the proposed Optivote algorithm with baseline schemes for distributed learning to examine its effectiveness.

\subsection{Experimental Setup and Dataset }

\emph{1) Scenario Setting:} For our simulations, we consider a typical distributed learning network consisting of one aggregation satellite and $\emph{M}$ space nodes. In each communication round, the AS randomly selects $\emph{m}$ active space nodes to participate in the distributed training procedure. The selected nodes perform local learning based on their own datasets and periodically upload the local gradients to the AS over FSO links for aggregation, while the AS broadcasts the updated global model to the nodes in the next round. 

\begin{figure}
        \centering     \includegraphics[width=1\linewidth]{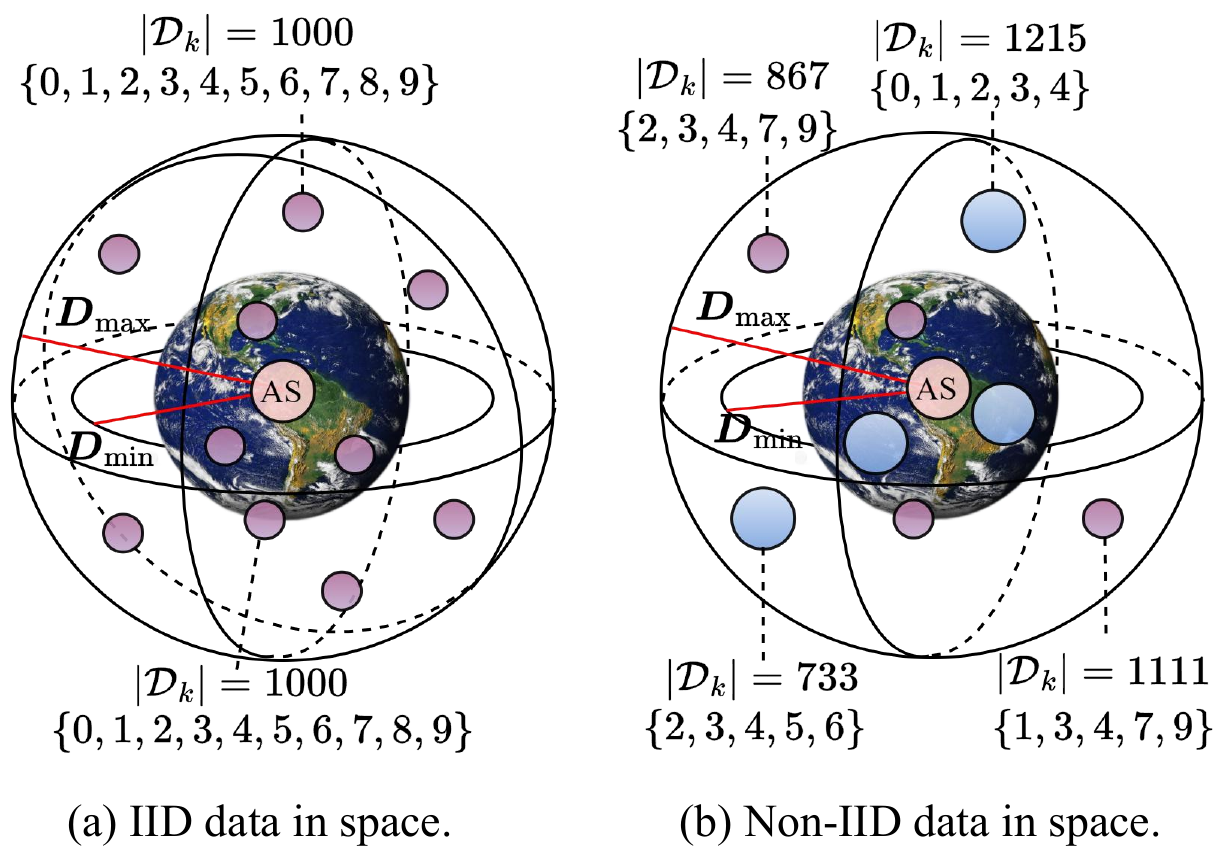}
        \caption{IID versus non-IID data considered for the detailed numerical analyses. Satellite nodes are randomly distributed inside the sphere, ranging $500\sim2000$ KM. (a): All space nodes have the same data samples for 10 different digits on their locations. (b): The available digits at the space nodes change based on their locations in space.}
        \label{fig3}
\end{figure}

\emph{2) Topology and Inter-satellite FSO Channel Model:}
The $M$ space nodes are randomly deployed in a three-dimensional spherical region centered at the aggregation satellite. 
Specifically, the link distance between node $m$ and the AS, denoted by $d_k$, is independently drawn from a bounded range $[d_{\min}, d_{\max}]$, where $d_{\min}$ = $500$ km and $d_{\max}=2000$ km. 
Moreover, we focus on inter-satellite optical links, where atmospheric turbulence is negligible. 
Therefore, the FSO channel impairments are modeled by only geometric path loss and pointing error.

\emph{3) Dataset and Neural Network Architecture}
To evaluate performance of the proposed Algorithm, we simulate the image classification task on MNIST and CIFAR-10 data sets. FL performance is evaluated by test accuracy, i.e., the number of correctly classified test images to the size of the test set ratio. During training, we apply standard data augmentation techniques, including random cropping and horizontal flipping, to enhance generalization. Based on CIFAR-10 and Tiny-ImageNet, we design the different model backbones.

\emph{3) Dataset and Neural Network Architecture:} To evaluate performance of the proposed Algorithm, we simulate the image classification task on MNIST and CIFAR-10 data sets. FL performance is evaluated by test accuracy, i.e., the number of correctly classified test images to the size of the test set ratio. During training, we apply standard data augmentation techniques, including random cropping and horizontal flipping, to enhance generalization. Based on CIFAR-10 and Tiny-ImageNet, we design the different model backbones.

\emph{4) Compared Methods:} We evaluate OptiVote against other state-of-the-art baselines, described as follows:
\begin{itemize}
\item \textbf{OBDA without TCI:} This scheme uses distributed learning with AirComp framework in \cite{9272666}, but disables the timing/synchronization compensation module, i.e., the transmission coefficient information (TCI). 
\item \textbf{OBDA with TCI:} This method \cite{9272666} incorporates TCI to mitigate the impact of timing synchronization mismatch in over-the-air aggregation. We adopt the same training protocol and denote this baseline as OBDA with TCI.
\item \textbf{FSK-MV:} This scheme applies FSK-based signaling to aggregate the one-bit signs from multiple clients, offering enhanced robustness against channel impairments \cite{9641940}.
\item \textbf{FedAvg-AirComp:} This baseline operates the standard FedAvg and combines AirComp strategy to process the upward transmission of model parameters.
\end{itemize}

To ensure fair comparisons, we standardize the experimental environment across all methods. Furthermore, all models share an identical neural network backbone, subject to equal constraints on computational complexity and memory usage, simulating deployment on resource-limited edge devices.

\subsection{Experimental Results}

In the experiments, we evaluate the learning performance in a typical AS and multiple space nodes distributed learning network. In each communication round, the AS randomly selects $\emph{m}=4$ active nodes from a population of 
$\emph{M}=100$ nodes. Each selected node performs local training on its private dataset for $\emph{E}=5$ local epochs with batch size $\emph{B}=64$, and then uploads its update to the AS through an inter-satellite FSO uplink for aggregation. The AS broadcasts the updated global model to all nodes for the next round.

\emph{1) Task Performance on Distributed Network:} In Fig.~\ref{fig4a} and Fig.~\ref{fig4b}, we report the test accuracy under MNIST both IID and Non-IID data partitions in the considered AS-space nodes FL network with $\emph{M}=100$ nodes and $\emph{m}=4$ active nodes per round. Specifically, under IID data, the one-bit/voting-based schemes rise extremely fast and quickly enter a high-accuracy regime. Note that OptiVote shows the steepest early-stage climb and reaches the top level within roughly the first few tens of communication rounds, after which it remains highly stable and achieves the best steady-state performance. 
By adaptively allocating transmit power according to the importance of local updates, the AS in our scheme can effectively enlarge the vote margin at the receiver and suppress the detrimental influence of weak or misaligned contributions, which reduces sign decision errors especially near convergence.
Consistent with this, other baseline (such as FSK-MV) closely tracks the proposed scheme but converges to a slightly lower region. Also, the inset shows a small yet persistent gap, suggesting that the proposed power control not only accelerates convergence but also improves the final-stage reliability. OBDA with TCI and OBMA without TCI also reach the high-accuracy region in the IID case. However, the TCI-enabled variant tends to be smoother and marginally higher in the late stage, indicating that when egde gradients are well aligned, synchronization compensation mainly reduces residual decoding errors rather than changing the overall convergence trend. In contrast, FedAvg-AirComp exhibits fundamentally different behavior. Although IID data typically favors FedAvg, analog superposition of real-valued updates over the FSO uplink introduces aggregation distortion and bias that can prevent effective learning precedure.

\begin{figure*}[t]
        \centering
        \subfigure[Task performance on IID case in MNIST dataset.]{{\label{fig4a}}\includegraphics[width=0.49\linewidth]{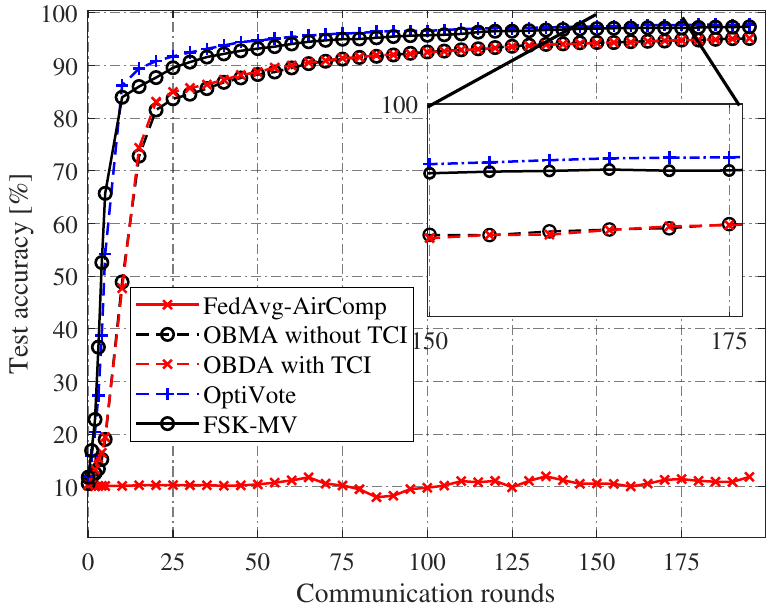}}
        \subfigure[Task performance on Non-IID case in MNIST dataset.]{{\label{fig4b}}\includegraphics[width=0.49\linewidth]{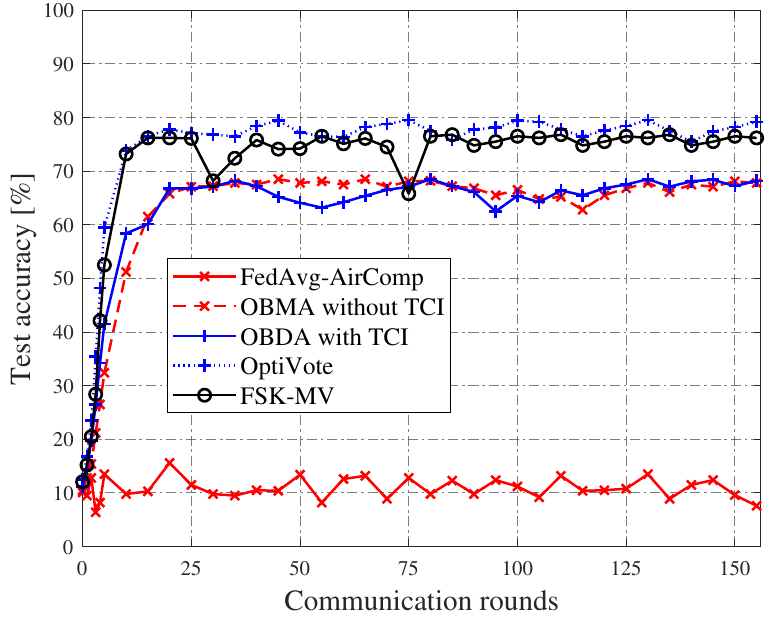}}
        \caption{MNIST test accuracy versus communication rounds in the considered AS-space-nodes FL network over inter-satellite FSO uplinks. We compare FedAvg-AirComp, OBMA without TCI, OBDA with TCI, OptiVote, and FSK-MV under (a) IID and (b) Non-IID data partitions.}
        \label{fig4}
\end{figure*}

\begin{figure*}[t]
        \centering
        \subfigure[Model loss on IID case in MNIST dataest.]{{\label{fig5a}}\includegraphics[width=0.49\linewidth]{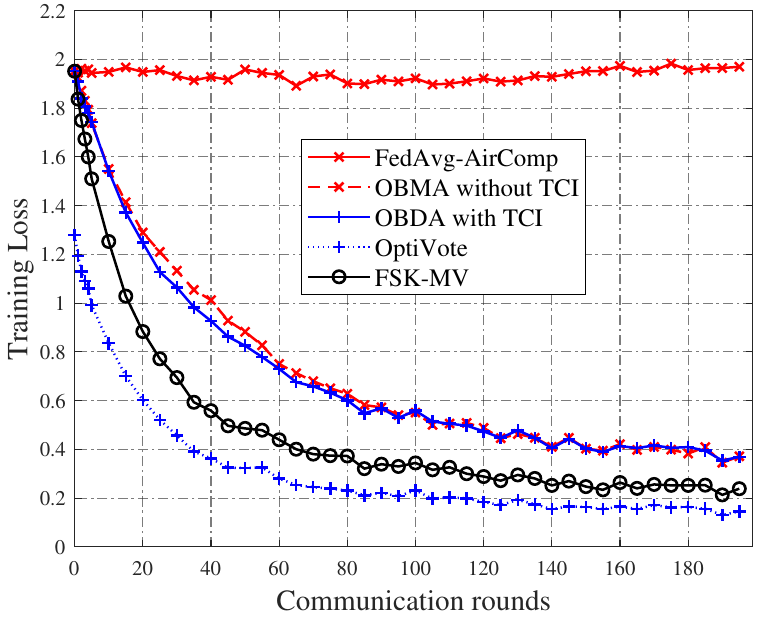}}
        \subfigure[Model loss on Non-IID case in MNIST dataest.]{{\label{fig5b}}\includegraphics[width=0.49\linewidth]{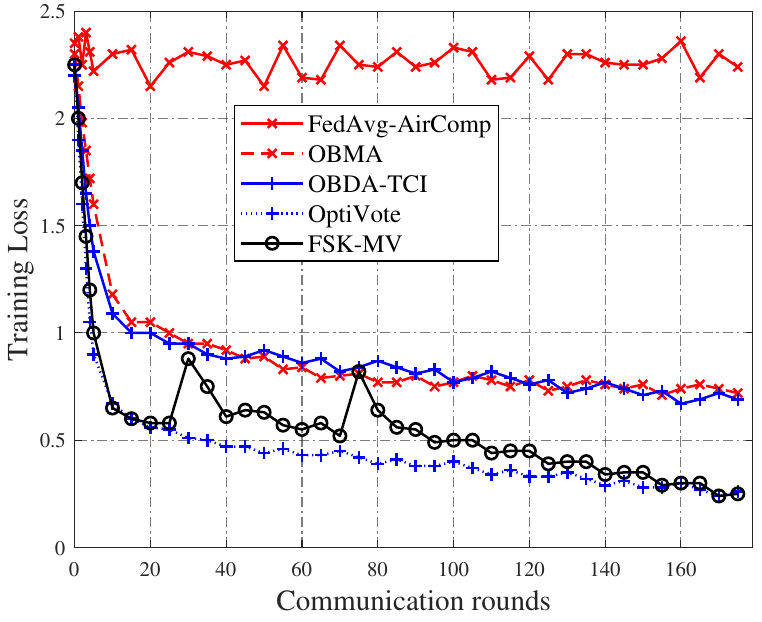}}
        \caption{MNIST training loss versus communication rounds under the same experimental setting as Fig. 4. We report the convergence behavior of FedAvg-AirComp, OBMA without TCI, OBDA with TCI, OptiVote, and FSK-MV for (a) IID and (b) Non-IID data partitions, illustrating both the convergence speed and the stability of the learning process over inter-satellite FSO aggregation.}
        \label{fig5}
\end{figure*}

Regarding the Non-IID setting in Fig.~\ref{fig4b}, all methods display a slower and more fluctuating trajectory than in the IID case, because each round’s randomly selected $m=4$ nodes may carry highly unbalanced label distributions, leading to stronger gradient disagreement. Even in this more challenging set, the proposed scheme remains robust and consistently stays in the highest accuracy band with relatively small oscillations, while FSK-MV remains the second-best curve but experiences slightly deeper dips in some rounds. Notably, OBDA with TCI generally shows fewer and shallower drops than its non-TCI counterpart, consistent with the fact that synchronization-induced distortion is more likely to cause decoding errors under heterogeneous updates. Meanwhile, FedAvg-AirComp continues to hover near the chance level with irregular small spikes, implying that the combination of biased participation and analog aggregation distortion makes the aggregated update highly unreliable. 
Overall, these results demonstrate that OptiVote achieves the best trade-off among fast convergence, high steady-state accuracy, and robustness to data heterogeneity, largely due to the importance-aware power control that enhances decision reliability without relying on explicit CSI.

\emph{2) Realistic Convergence Analysis:} In the IID case of Fig.~\ref{fig5a}, the proposed schemes exhibit clear and stable loss decay. OptiVote drops from the initial $2$ level to below $0.3$ within roughly the first few tens of communication rounds and further converges to the lowest loss floor about $0.15$. FSK-MV follows closely but stabilizes at a slightly higher floor (roughly $0.2\sim0.25$). However, OBDA-TCI and OBMA without TCI converge slower and remain higher around $0.35\sim0.45$. In contrast, FedAvg-AirComp essentially fails to optimize, and the training loss stays around $2.0$ throughout training, indicating that the analog superposition over the FSO uplink introduces a persistent aggregation distortion/bias that overwhelms the true descent direction. Further, in the Non-IID case on Fig.~\ref{fig5a}, the trajectories become more oscillatory and the steady-state loss is higher, which indicates that data heterogeneity increases gradient disagreement and makes the vote margin more fragile. Still, FedAvg-AirComp stays near $2.2\sim2.4$, confirming its strong sensitivity to analog aggregation distortion. Notably, OptiVote maintains the most stable low-loss band (about $0.35\sim0.45$ after convergence), while OBDA/OBMA plateau higher (around $0.7\sim0.9$), and FSK-MV shows occasional spikes, i.e., a typical signature of bursty decoding errors when the margin is small. These observations directly support our main schem, i.e., the importance-aware dynamic power control increases the effective receive margin for influential contributors, thereby reducing the dominant vote-error probability term, which translates into faster convergence and a lower loss floor, especially under Non-IID where errors are more likely to flip marginal votes.


\section{Conclusion}
This paper targets communication-efficient and robust federated learning over inter-satellite FSO uplinks, where coherent AirComp is hard to sustain due to phase-related impairments. We proposed OptiVote, a non-coherent FSO AirComp framework that integrates signSGD with majority-vote aggregation and PPM-based signaling, enabling simultaneous uplink aggregation via energy accumulation without requiring precise phase synchronization. To further combat aggregation bias caused by heterogeneous FSO links, we developed an importance-aware, CSI-free dynamic power control mechanism that balances the received energies without extra signaling overhead. We complemented the design with theoretical characterization of the aggregate decision error probability under statistical FSO channels and established convergence guarantees for non-convex objectives under non-coherent MV aggregation with communication impairments. Extensive experiments under both IID and Non-IID partitions demonstrated that OptiVote consistently improves learning accuracy and communication efficiency over representative baselines, supporting scalable and resilient in-orbit intelligence for future communication-constrained space data centers.

\begin{appendices}

\section{Proof of Lemma 1}\label{app:l1}

Consider the received energy statistic contributed by space node $m$ at round $n$,
\begin{equation}
e_m^{(n)} \triangleq C_R P_m^{(n)} I_m^{(n)} \sqrt{E_s} + \epsilon_m^{(n)},
\label{eq:em_def_rewrite}
\end{equation}
where the concrete noise contribution to the expected energy statistic is quantified by its average power, i.e., $\sigma_n^2$). 
Taking the expectation yields
\begin{equation}
\mathbb{E}\!\left[e_m^{(n)}\right] = C_R\sqrt{E_s}\,\mathbb{E}\!\left[P_m^{(n)} I_m^{(n)}\right]+ \sigma_n^2.
\end{equation}
Since $P_m^{(n)}$ is updated by the proposed importance-aware scheme, while the channel gain $I_m^{(n)}$ is determined by the constellation topology and the platform pointing jitter, we adopt the standard independence assumption $P_m^{(n)}\perp I_m^{(n)}$, which gives
\begin{equation}
\mathbb{E}\!\left[P_m^{(n)} I_m^{(n)}\right]
=\mathbb{E}\!\left[P_m^{(n)}\right]\mathbb{E}\!\left[I_m^{(n)}\right] =P_{\text{avg}}\lambda,
\end{equation}
where $\bar{P}\triangleq\mathbb{E}[P_m^{(n)}]$ and $\lambda\triangleq\mathbb{E}[I_m^{(n)}]$. Thus, the effective average received energy per satellite is
\begin{equation}
\vartheta \triangleq \mathbb{E}[C_R P_m^{(n)} I_m^{(n)} \sqrt{E_s}]= C_R\sqrt{E_s}\,P_{\text{avg}}\lambda,
\end{equation}
which proves Eq.(26).

Let $\mathcal{M}_i^+$ (resp. $\mathcal{M}_i^-$) denote the set of satellites voting for $+1$ (resp. $-1$), with $|\mathcal{M}_i^+|=M_i^+$ and $|\mathcal{M}_i^-|=M_i^-$. The accumulated energies at the aggregation station are defined as
\begin{equation}
e_i^+ \triangleq \sum_{m\in\mathcal{M}_i^+} e_m^{(n)},\qquad
e_i^- \triangleq \sum_{m\in\mathcal{M}_i^-} e_m^{(n)}.
\end{equation}
By the linearity of expectation, we obtain
\begin{equation}
\mu_i^+ \triangleq \mathbb{E}[e_i^+]=\sum_{m\in\mathcal{M}_i^+}\mathbb{E}[e_m^{(n)}]=M_i^+\vartheta + \sigma_n^2,
\end{equation}
and
\begin{equation}
\mu_i^- \triangleq \mathbb{E}[e_i^-]=\sum_{m\in\mathcal{M}_i^-}\mathbb{E}[e_m^{(n)}]=M_i^-\vartheta + \sigma_n^2.
\end{equation}
respectively. Then, we offer the specific proof on $\lambda=\mathbb{E}[I_m^{(n)}]$. Under the channel decomposition $I_m^{(n)}=h_{l,m}h_{p,m}$ and the independence between geometric path loss and pointing jitter, we have $\lambda=\mathbb{E}[h_l]\mathbb{E}[h_p]$. For satellites uniformly distributed in a 3D spherical shell with $D\in[D_{\min},D_{\max}]$, the specific PDF is
\begin{equation}
f_D(d)=\frac{3d^2}{D_{\max}^3-D_{\min}^3},\quad d\in[D_{\min},D_{\max}],
\end{equation}
and with $h_l=C_{\text{FSPL}}/D^2$, the geometric efficiency admits the closed-form
\begin{equation}
\begin{aligned}
\mathbb{E}[h_l]
=&\int_{D_{\min}}^{D_{\max}} \frac{C_{\text{FSPL}}}{d^2}\cdot \frac{3d^2}{D_{\max}^3-D_{\min}^3}\,\mathrm{d}d
\\=&\frac{3C_{\text{FSPL}}(D_{\max}-D_{\min})}{D_{\max}^3-D_{\min}^3}.    
\end{aligned}
\end{equation}
where $C_{\text{FSPL}} \triangleq (\lambda_{\text{opt}}/{4\pi})^2$ is the free-space path loss constant determined by the optical carrier wavelength $\lambda_{\text{opt}}$. $D_{\min}$ and $D_{\max}$ denote the minimum and maximum inter-satellite distances of the spherical cluster, respectively. Regarding the pointing errors, $A_0$ represents the maximum fraction of collected power at zero radial displacement. 

Note that these terms are fixed system constants determined by the satellite orbit design and optical transceiver hardware. Moreover, for the standard zero-boresight jitter model, the pointing loss $h_p\in(0,A_0]$ follows
\begin{equation}
f_{h_p}(h)=\frac{\xi^2}{A_0^{\xi^2}}h^{\xi^2-1},\quad 0<h\le A_0,
\end{equation}
which yields
\begin{equation}
\mathbb{E}[h_p]
=\int_{0}^{A_0} h\,f_{h_p}(h)\,\mathrm{d}h
=A_0\frac{\xi^2}{\xi^2+1}.
\end{equation}
Combining the above results gives $\lambda=\mathbb{E}[h_l]\mathbb{E}[h_p]$ in Eq.(\ref{eq27}). The proof is complete.

\section{Proof of Lemma 2}\label{app:l2}

Regarding this bounds on stochasticity-induced error, we mainly deal with the term of $P_{i}^{\mathrm{err}}$. Through Section IV-C, the following treatment is available for Eq.(\ref{eq39}), for all $m$. This implies that
\begin{equation}\label{eq53}
\emph{P}_\emph{i}^{\text{err}}=\sum_{\emph{M}_{i}^{+}=0}^{\emph{M}} \mathbb{P}\left[\text{sign}\left(\Delta_{\emph{i}}^{(\emph{n})}\right) \neq 1 \mid \emph{Z}=\emph{M}_{\emph{i}}^{\textbf{\texttt{+}}}\right]\mathbb{P}\left[\emph{Z}=\emph{M}_{\emph{i}}^{\textbf{\texttt{+}}}\right].
\end{equation}

According to the properties of Bernoulli distribution, the second term on Eq.(\ref{eq53}) can be expressed as follow:
\begin{equation}
\mathbb{P}\left[\emph{Z}=\emph{M}_{\emph{i}}^{\textbf{\texttt{+}}}\right]=\left(\begin{array}{c}
\emph{M} \\\emph{M}_{\emph{i}}^{\textbf{\texttt{+}}}\end{array}\right) \emph{p}_\emph{i}^{\emph{M}_{\emph{i}}^{\textbf{\texttt{+}}}} \emph{q}_\emph{i}^{\emph{M}\texttt{-}\emph{M}_{\emph{i}}^{\textbf{\texttt{+}}}}.
\end{equation}

To calculate $\mathbb{P}[\text{sign}(\Delta_{\emph{i}}^{(\emph{n})}) \neq 1 \mid \emph{Z}=\emph{M}_{\emph{i}}^{\textbf{\texttt{+}}}]$, we utilize the statistical properties of the accumulated energies. Based on Assumption 6, the accumulated energies $\emph{e}_{\emph{i}}^{\textbf{\texttt{+}}}$ and $\emph{e}_{\emph{i}}^{\textbf{\texttt{-}}}$ follow exponential distributions. The conditional error probability is determined by the ratio of their means:
\begin{equation}\label{eq55}
\begin{aligned}
\mathbb{P}\left[\text{sign}\left(\Delta_{\emph{i}}^{(\emph{n})}\right) \neq 1 \mid \emph{Z}=\emph{M}_{\emph{i}}^{\textbf{\texttt{+}}}\right] &= \mathbb{P}\left(\emph{e}_\emph{i}^\textbf{\texttt{+}} < \emph{e}_\emph{i}^\textbf{\texttt{-}}\mid \emph{Z}=\emph{M}_{\emph{i}}^{\textbf{\texttt{+}}}\right) \\
&= \frac{\mu_\emph{i}^\textbf{\texttt{-}}}{\mu_\emph{i}^\textbf{\texttt{+}} \texttt{+} \mu_\emph{i}^\textbf{\texttt{-}}} \\
&=\frac{(\emph{M}\texttt{-}\emph{M}_{\emph{i}}^{\textbf{\texttt{+}}}) \texttt{+} 1/\xi}{\emph{M} \texttt{+} 2/\xi}.
\end{aligned}
\end{equation}
where $\beta \triangleq \vartheta / \sigma_n^2$. Substituting Eq.(\ref{eq55}) and the Binomial PMF back into the total probability formula, we obtain:
\begin{equation}
\emph{P}_{\emph{i}}^{\mathrm{err}} = \sum_{\emph{M}_{\emph{i}}^{\textbf{\texttt{+}}}=0}^{\emph{M}}\frac{(\emph{M}\texttt{-}\emph{M}_{\emph{i}}^{\textbf{\texttt{+}}}) \texttt{+} 1/\xi}{\emph{M} \texttt{+} 2/\xi} \left(\begin{array}{c}
\emph{M} \\\emph{M}_{\emph{i}}^{\textbf{\texttt{+}}}\end{array}\right) \emph{p}_\emph{i}^{\emph{M}_{\emph{i}}^{\textbf{\texttt{+}}}} \emph{q}_\emph{i}^{\emph{M}\texttt{-}\emph{M}_{\emph{i}}^{\textbf{\texttt{+}}}}.
\end{equation}

To solve this issue, we exploit the linearity of the summation and split the right-hand side into two distinct terms, representing the data noise impact and channel noise impact, respectively:
\begin{equation}\label{eq57}
\begin{aligned}
\emph{P}_{\emph{i}}^{\mathrm{err}} = &\underbrace{\sum_{\emph{M}_{\emph{i}}^{\textbf{\texttt{+}}}=0}^{\emph{M}}\frac{(\emph{M}\texttt{-}\emph{M}_{\emph{i}}^{\textbf{\texttt{+}}})\xi}{\emph{M} \xi\texttt{+} 2} \left(\begin{array}{c}
\emph{M} \\\emph{M}_{\emph{i}}^{\textbf{\texttt{+}}}\end{array}\right) \emph{p}_\emph{i}^{\emph{M}_{\emph{i}}^{\textbf{\texttt{+}}}} \emph{q}_\emph{i}^{\emph{M}\texttt{-}\emph{M}_{\emph{i}}^{\textbf{\texttt{+}}}}}_{\emph{T}_1} \\+& \underbrace{\sum_{\emph{M}_{\emph{i}}^{\textbf{\texttt{+}}}=0}^{\emph{M}}\frac{1}{\emph{M} \xi\texttt{+} 2} \left(\begin{array}{c}
\emph{M} \\\emph{M}_{\emph{i}}^{\textbf{\texttt{+}}}\end{array}\right) \emph{p}_\emph{i}^{\emph{M}_{\emph{i}}^{\textbf{\texttt{+}}}} \emph{q}_\emph{i}^{\emph{M}\texttt{-}\emph{M}_{\emph{i}}^{\textbf{\texttt{+}}}}}_{\emph{T}_2}.
\end{aligned}
\end{equation}

Hence, by using Eq.(\ref{eq57}) and the properties of binomial coefficients, the first term can be obtained as:
\begin{equation}
\begin{aligned}
\emph{T}_1 &= \frac{\xi}{\emph{M} \xi\texttt{+} 2} \sum_{\emph{M}_{\emph{i}}^{\textbf{\texttt{+}}}=0}^{\emph{M}} (\emph{M}\texttt{-}\emph{M}_{\emph{i}}^{\textbf{\texttt{+}}}) \left(\begin{array}{c}
\emph{M} \\\emph{M}_{\emph{i}}^{\textbf{\texttt{+}}}\end{array}\right) \emph{p}_\emph{i}^{\emph{M}_{\emph{i}}^{\textbf{\texttt{+}}}} \emph{q}_\emph{i}^{\emph{M}\texttt{-}\emph{M}_{\emph{i}}^{\textbf{\texttt{+}}}} \\
&= \frac{\xi}{\emph{M} \xi\texttt{+} 2} \cdot \mathbb{E}[\emph{M}\texttt{-}\emph{M}_{\emph{i}}^{\textbf{\texttt{+}}}].
\end{aligned}
\end{equation}
Since $\emph{Z} \sim \mathcal{B}(\emph{M}, \emph{p}_\emph{i})$, $\mathbb{E}[\emph{M}\texttt{-}\emph{M}_{\emph{i}}^{\textbf{\texttt{+}}}] = \emph{M}(1\texttt{-}\emph{p}_\emph{i}) = \emph{M} \emph{q}_\emph{i}$. Thus, 
\begin{equation}
\emph{T}_1 = \frac{\emph{M} \emph{q}_\emph{i}\xi}{\emph{M} \xi\texttt{+} 2} = \frac{\emph{M} \emph{q}_\emph{i}}{\emph{M}\texttt{+} 2/\xi}.
\end{equation}

Similarly, we extract the constant factor to calculate the second term
\begin{equation}
\begin{aligned}
\emph{T}_2 &=  \frac{1}{\emph{M} \xi\texttt{+} 2} \sum_{\emph{M}_{\emph{i}}^{\textbf{\texttt{+}}}=0}^{\emph{M}}\left(\begin{array}{c}
\emph{M} \\\emph{M}_{\emph{i}}^{\textbf{\texttt{+}}}\end{array}\right) \emph{p}_\emph{i}^{\emph{M}_{\emph{i}}^{\textbf{\texttt{+}}}} \emph{q}_\emph{i}^{\emph{M}\texttt{-}\emph{M}_{\emph{i}}^{\textbf{\texttt{+}}}} \\
&= \frac{1}{\emph{M} \xi\texttt{+} 2} \cdot 1 = \frac{1/\xi}{\emph{M}\texttt{+} 2/\xi}.
\end{aligned}
\end{equation}

Combining $\emph{T}_1$ and $\emph{T}_2$, we can obtain
\begin{equation}
\emph{P}_{\emph{i}}^{\mathrm{err}} = \emph{T}_1 \texttt{+} \emph{T}_2 = \frac{\emph{M} \emph{q}_\emph{i}}{\emph{M}\texttt{+} 2/\xi} \texttt{+} \frac{1/\xi}{\emph{M}\texttt{+} 2/\xi}.
\end{equation}

To proceed with, we utilize the the specific bound on $\emph{q}_{\emph{i}} \leq ({\sqrt{2} \alpha_{\emph{i}}})/({3\left|\text{g}_{\emph{i}}^{(\emph{n})}\right| \sqrt{\emph{d}_{\emph{b}}}})$ shown as in Lemma 3.

The proof process relies on the properties of certain probability distributions, which is captured in Appendix \ref{app:l3}.

Under symmetry assumption and derivations in Lemma 3, we combine the upper bound on $\text{q}_{\emph{i}}$ with Eq.(\ref{eq62}) to obtain

\begin{equation}
\emph{P}_{\emph{i}}^{\mathrm{err}} \leq \frac{\emph{M}({\sqrt{2} \alpha_{\emph{i}}})/({3\left|\text{g}_{\emph{i}}^{(\emph{n})}\right| \sqrt{\emph{d}_{\emph{b}}}})}{\emph{M}\texttt{+} 2/\xi} + \frac{1/\xi}{\emph{M}\texttt{+} 2/\xi}.
\end{equation}
This completes the proof. 

\section{Proof of Lemma 3} \label{app:l3}

According to assumption in \cite{bernstein2018signsgd} and section IV-A, for a unimodal symmetric random variable
$\emph{Y}$ with mean $\varphi$ and variance $\alpha^{2}$, the following Gauss’ inequality holds:
\begin{equation}
\begin{aligned}
\mathbb{P}[|\emph{Y}-\varphi|>y] \leq\left\{\begin{array}{ll}
\frac{4}{9} \frac{\alpha^{2}}{y^{2}}, & \text { if } \frac{y}{\alpha}>\frac{2}{\sqrt{3}} \\
1-\frac{y}{\sqrt{3} \alpha}, & \text { otherwise. }
\end{array}\right.
\end{aligned}
\end{equation}

Then applying symmetry followed by Gauss’ inequality, the failure probability can be obtained by
\begin{equation}
\begin{aligned}
\mathbb{P}&\left[\text{sign}\left(\tilde{\text{g}}_{\emph{m,i}}^{(\emph{n})}\right)\neq\text{sign}\left(\text{g}_{\emph{i}}^{(\emph{n})}\right)\right]=\mathbb{P}\left[\tilde{\text{g}}_{\emph{m,i}}^{(\emph{n})}-\text{g}_{\emph{i}} \geq|\text{g}_{\emph{i}}|\right] \\
& =\frac{1}{2} \mathbb{P}\left[|\tilde{\text{g}}_{\emph{m,i}}^{(\emph{n})}-\text{g}_{\emph{i}}| \geq|\text{g}_{\emph{i}}|\right] \\
& \leq\left\{\begin{array}{ll}
\frac{2}{9} \frac{\alpha_{\emph{i}}^{2}}{\emph{d}_{\emph{b}}|\text{g}_{\emph{i}}^{(\emph{n})}|^{2}}, & \text { if } \frac{|\text{g}_{\emph{i}}^{(\emph{n})}|}{\alpha_{i} / \sqrt{\emph{d}_{\emph{b}}}}>\frac{2}{\sqrt{3}} \\
\frac{1}{2}-\frac{|\text{g}_{\emph{i}}^{(\emph{n})}|}{2 \sqrt{3} \alpha_{i} / \sqrt{\emph{d}_{\emph{b}}}}, & \text { otherwise, }
\end{array}\right.
\end{aligned}
\end{equation}
which is in all cases less than $1/2$. Eventually, we complete the proof, which is used to infer Lemma 3. 

\section{Proof of Theorem 1} \label{app:T1}

To begin with, we derive the target based on the noise introduced by data randomness according to Assumption 2. For this process, we decompose it into the data requiring analysis and the error caused by channel randomness. Thus, we obtain:
\begin{equation}
\begin{aligned}
\emph{F}(\mathbf{w}^{(\emph{n+1})}) & -\emph{F}(\mathbf{w}^{(\emph{n})}) \leq \mathbf{g}^{(\emph{n})^{\mathrm{\emph{T}}}}\left(\mathbf{w}^{(\emph{n+1})}-\mathbf{w}^{(\emph{n})}\right)\\&
+\frac{1}{2} \sum_{\emph{i}=1}^{\emph{q}} \emph{L}_{\emph{i}}\left(\emph{w}_{\emph{i}}^{(\emph{n+1})}-\emph{w}_{\emph{i}}^{(\emph{n})}\right)^{2}. 
\end{aligned}
\end{equation}
Then, we can make a substitution with $(\hat{\emph{v}}_{\emph{i}}^{(\emph{n})})^2=1$, whether it is $+1$ or $-1$. Thus, we have
\begin{equation}
\begin{aligned}
\emph{F}(\mathbf{w}^{\emph{(n+1)}}) & -\emph{F}(\mathbf{w}^{\emph{(n)}}) \leq-\eta \mathbf{g}^{\emph{(n)}^{\emph{T}}}\mathbf{\hat{v}}^{\emph{(n)}}+\frac{1}{2} \sum_{\emph{i}=1}^{\emph{q}} \emph{L}_{\emph{i}}(\eta\tilde{\emph{v}}_{\emph{i}}^{\emph{(n)}})^2 \\
= & -\eta\left\|\mathbf{g}^{\emph{(n)}}\right\|_{1}\text{sign}\left(\Delta_{\emph{i}}^{\emph{(n)}}\right)+\frac{\eta^{2}}{2}\|\mathbf{L}\|_{1}. \\&
\end{aligned}
\end{equation}

As $\text{sign}(\cdot)$ can not be determined, so the term have a randomness error. Thus, we then proceed to obtain
\begin{equation}
\begin{aligned}
\emph{F}(\mathbf{w}^{\emph{(n+1)}}) & -\emph{F}(\mathbf{w}^{\emph{(n)}}) \leq-\eta \mathbf{g}^{\emph{(n)}^{\emph{T}}}\mathbf{\hat{v}}^{\emph{(n)}}+\frac{\eta^{2}}{2}\|\mathbf{L}\|_{1} \\
= & -\eta\left\|\mathbf{g}^{\emph{(n)}}\right\|_{1}+\frac{\eta^{2}}{2}\|\mathbf{L}\|_{1} \\
& +2 \eta \sum_{\emph{i}=1}^{\emph{q}}|\text{g}_{\emph{i}}^{\emph{(n)}}| \mathbb{I}\left[\text{sign}\left(\Delta_{\emph{i}}^{\emph{(n)}}\right) \neq \text{sign}\left(\text{g}_{\emph{i}}^{\emph{(n)}}\right)\right].
\end{aligned}
\end{equation}

Thus, we can further obtain 
\begin{equation}
\begin{aligned}
\mathbb{E}\left[\emph{F}(\mathbf{w}^{\emph{(n+1)}})\right. & \left.-\emph{F}(\mathbf{w}^{\emph{(n)}}) \mid \mathbf{w}^{\emph{(n)}}\right] \leq-\eta\left\|\mathbf{g}^{\emph{(n)}}\right\|_{1}+\frac{\eta^{2}}{2}\|\mathbf{L}\|_{1} \\
+ & \underbrace{2 \eta \sum_{\emph{i}=1}^{\emph{q}}|\text{g}_{\emph{i}}^{\emph{(n)}}| \underbrace{\mathbb{P}\left[\text{sign}\left(\Delta_{\emph{i}}^{\emph{(n)}}\right) \neq \text{sign}\left(\text{g}_{\emph{i}}^{\emph{(n)}}\right)\right]}_{\triangleq \emph{P}_\emph{i}^{\text{err}}}}_{\text {Stochasticity-induced error }}.
\end{aligned}
\end{equation}

Accordingly, based on Lemma 3 and several definitions of Theorem 1, an upper bound on the stochasticity-induced error can be represented by the proof related to Appendix A as follows:
\begin{equation}
\begin{aligned}
\sum_{\emph{i}=1}^{\emph{q}}|\text{g}_{\emph{i}}^{\emph{(n)}}| \emph{P}_\emph{i}^{\text{err}} \leq
&\sum_{\emph{i}=1}^{\emph{q}} \frac{\emph{M}\cdot|\text{g}_{\emph{i}}^{\emph{(n)}}|}{\emph{M}\texttt{+} 2/\xi} \cdot \frac{\sqrt{2}|\alpha_{\emph{i}}^{\emph{(n)}}|}{3|\text{g}_{\emph{i}}^{\emph{(n)}}| \sqrt{\emph{d}_{\emph{b}}}} \\+&\sum_{\emph{i}=1}^{\emph{q}}|\text{g}_{\emph{i}}^{\emph{(n)}}| \frac{1/\xi}{\emph{M}\texttt{+} 2/\xi}.
\end{aligned}
\end{equation}

Thus, we can obtain as follow:
\begin{equation}
\begin{aligned}
\sum_{\emph{i}=1}^{\emph{q}}|\text{g}_{\emph{i}}^{\emph{(n)}}| \emph{P}_\emph{i}^{\text{err}}\leq
\frac{\emph{M}}{\emph{M}\texttt{+} 2/\xi} \cdot \frac{\sqrt{2}}{3 \sqrt{\emph{d}_{\emph{b}}}}\|\boldsymbol{\alpha}\|_{1}+\frac{1/\xi}{\emph{M}\texttt{+} 2/\xi}\|\boldsymbol{\mathbf{g}}^{\emph{(n)}}\|_{1}.
\end{aligned}
\end{equation}

Then, we perform the following operation under the Assumptions 1-4, as
\begin{equation}\label{eq51}
\begin{aligned}
\emph{F}&(\mathbf{w}^{(0)})-\emph{F}^{*} \geq \emph{F}(\mathbf{w}^{(0)})-\emph{F}(\mathbf{w}^{\emph{(N)}})
\\=&\mathbb{E}\left[\sum_{\emph{n}=0}^{\emph{N-1}} \emph{F}(\mathbf{w}^{\emph{(n)}})-\emph{F}(\mathbf{w}^{\emph{(n+1)}})\right]\\
\geq & \mathbb{E}\left[\sum_{\emph{n}=0}^{\emph{N-1}}\left(\left(\eta-2 \eta \cdot \frac{1 / \xi}{\emph{M}\texttt{+} 2/\xi}\right)\|\boldsymbol{\mathbf{g}}^{\emph{(n)}}\|_{1}-\frac{\eta^{2}}{2}\|\mathbf{L}\|_{1} \right.\right. \\&\left.\left. -\frac{2 \eta \cdot \emph{M}}{\emph{M}\texttt{+} 2/\xi} \cdot \frac{\sqrt{2}}{3 \sqrt{\emph{d}_{\emph{b}}}} \cdot \| \boldsymbol{\alpha}\|_{1}\right)\right]
\\=& \mathbb{E}\left[\sum_{\emph{n}=0}^{\emph{N-1}}\left(\frac{\emph{M} \eta \cdot \|\boldsymbol{\mathbf{g}}^{\emph{(n)}}\|_{1}}{\emph{M}\texttt{+} 2/\xi} -\frac{\eta^{2}\|\mathbf{L}\|_{1}}{2}-\frac{2\sqrt{2} \emph{M} \eta \cdot \|\boldsymbol{\alpha}\|_{1}}{3 (\emph{M}\texttt{+} 2/\xi)\sqrt{\emph{d}_{\emph{b}}}}\right)\right].
\end{aligned}
\end{equation}

In order to derive the term of required convergence rate, we rearrange Eq.(\ref{eq51}) and use the expressions for
$\emph{d}_{\emph{b}}$ and $\eta$, while conducting a series of simplifications to obtain as follow:
\begin{equation}
\begin{aligned}
\mathbb{E} & \left[\frac{1}{\emph{N}} \sum_{\emph{n}=0}^{\emph{N-1}}\|\boldsymbol{\mathbf{g}}^{\emph{(n)}}\|_{1}\right]
\leq(1\texttt{+}\frac{2}{\emph{M} \xi}) \frac{\sqrt{\gamma}}{2 \sqrt{\emph{N}}}\cdot \sqrt{\| \mathbf{L} \|_{1}}\\\texttt{+}&
(1\texttt{+}\frac{2}{\emph{M} \xi}) \cdot \frac{\sqrt{\|\mathbf{L}\|_{1}}\sqrt{\emph{N}}}{\emph{N} \sqrt{\gamma}}(\emph{F}(\mathbf{w}^{(0)})
-\emph{F}^{*}(\mathbf{w}))\texttt{+}\frac{2 \sqrt{2} \sqrt{\gamma} \|\boldsymbol{\alpha}\|_{1}}{3 \sqrt{\emph{N}}}).
\end{aligned}
\end{equation}

This completes the proof. 

\end{appendices}

\bibliographystyle{IEEEtran} 
\bibliography{bib}
\end{document}